\documentclass{IEEEtran}
\usepackage{cite}

\usepackage{amsmath}
\usepackage{amssymb}
\usepackage{amsthm}
\usepackage{epsf}
\usepackage{epsfig}
\usepackage{multirow}
\usepackage{scalefnt}
\usepackage{subcaption}
\usepackage{siunitx}
\usepackage[normalem]{ulem}
\usepackage{xcolor}
\usepackage{enumerate}
\usepackage{empheq} 

\newtheorem{lemma}{Lemma}

 \def\cN{{\mathcal{N}}}

\def\given{\hspace{0.2mm}|\hspace{0.2mm}}

   \def\btheta{{\pmb{\theta}}}

  \def\bPhi{{\pmb{\Phi}}}

\def\bzeros{{\pmb{0}}}

 \def\bv{{\mathbf{v}}} \def\bw{{\mathbf{w}}} \def\bx{{\mathbf{x}}}
 \def\bz{{\mathbf{z}}}

\def\bA{{\mathbf{A}}} \def\bB{{\mathbf{B}}}  
 \def\bF{{\mathbf{F}}}  \def\bH{{\mathbf{H}}}
\def\bI{{\mathbf{I}}} \def\bJ{{\mathbf{J}}} \def\bK{{\mathbf{K}}} 
   \def\bP{{\mathbf{P}}}
\def\bQ{{\mathbf{Q}}} \def\bR{{\mathbf{R}}} \def\bS{{\mathbf{S}}} 
 \def\bV{{\mathbf{V}}}

\begin{document}

\title{Adaptive Temporal Decorrelation of State Estimates}
\author{
  \IEEEauthorblockN{Zachary Chance}\\
  \IEEEauthorblockA{
  \textit{MIT Lincoln Laboratory}\\
  Lexington, MA \\
  zachary.chance@ll.mit.edu}
  \thanks{
  DISTRIBUTION STATEMENT A. Approved for public release. Distribution is unlimited.~\(||\) This material is based upon work supported by the Department of the Air Force under Air Force Contract No. FA8702-15-D-0001. Any opinions, findings, conclusions or recommendations expressed in this material are those of the author(s) and do not necessarily reflect the views of the Department of the Air Force.~\(||\) Delivered to the U.S. Government with Unlimited Rights, as defined in DFARS Part 252.227-7013 or 7014 (Feb 2014). Notwithstanding any copyright notice, U.S. Government rights in this work are defined by DFARS 252.227-7013 or DFARS 252.227-7014 as detailed above. Use of this work other than as specifically authorized by the U.S. Government may violate any copyrights that exist in this work.

  © 2024 Massachusetts Institute of Technology.
  }
}

\maketitle

\begin{abstract}
Many commercial and defense applications involve multisensor, multitarget tracking, requiring the fusion of information from a set of sensors. An interesting use case occurs when data available at a central node (due to geometric diversity or retrodiction) allows for the tailoring of state estimation for a target. For instance, if a target is initially tracked with a maneuvering target filter, yet the target is clearly not maneuvering in retrospect, it would be beneficial at the fusion node to refilter that data with a non-maneuvering target filter. If measurements can be shared to the central node, the refiltering process can be accomplished by simply passing source measurements through an updated state estimation process. It is often the case for large, distributed systems, however, that only track information can be passed to a fusion center. In this circumstance, refiltering data becomes less straightforward as track states are linearly dependent across time, and the correlation needs to be properly accounted for before/during refiltering. In this work, a model-based temporal decorrelation process for state estimates with process noise will be studied. A decorrelation procedure will be presented based on a linear algebraic formulation of the problem, and process noise estimates will be created that ensure a conservative system state estimate. Numerical examples will be given to demonstrate the efficacy of the proposed algorithm.
\end{abstract}

\section{Introduction}
It is often advantageous to combine information from multiple sensors when tracking a scene with multiple targets. Fusion of multisensor data can greatly improve track accuracy by providing geometric diversity. For smaller sensor networks or circumstances where information exchange is not constrained, this is often approached by sharing measurement-level data, performing data association, and filtering using all sensor contributions. In many realistic situations, however, full sharing of measurements can be impractical due to number of sensors, geographic span of sensor network, or communications limitations. To reduce the amount of data transmitted, track states are sent to a central node for combination. Before any filtering of sensor data can be performed, the gathered track states have to be processed to deal with their correlation both in time and across sensors~\cite{BarShalom95, Chong18}. A comparison of measurement-based and track-based architectures is shown in Figure~\ref{fig:arch}.

\begin{figure}[h!]
\begin{subfigure}[b]{1\columnwidth}
  \centering
  \includegraphics[width=0.85\textwidth]{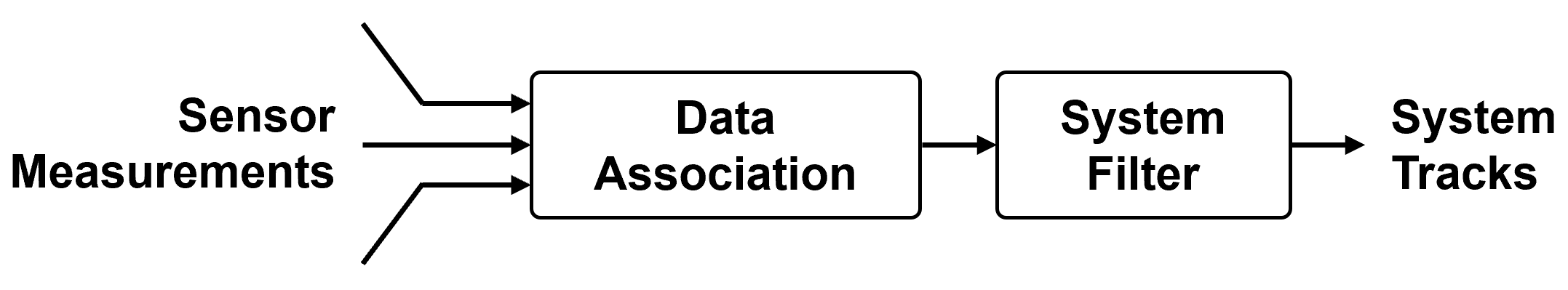}
  \caption{Multisensor, multitarget tracking with sensor measurements.}\vspace{3mm}
\end{subfigure}
\begin{subfigure}[b]{1\columnwidth}
  \centering
  \includegraphics[width=1.0\textwidth]{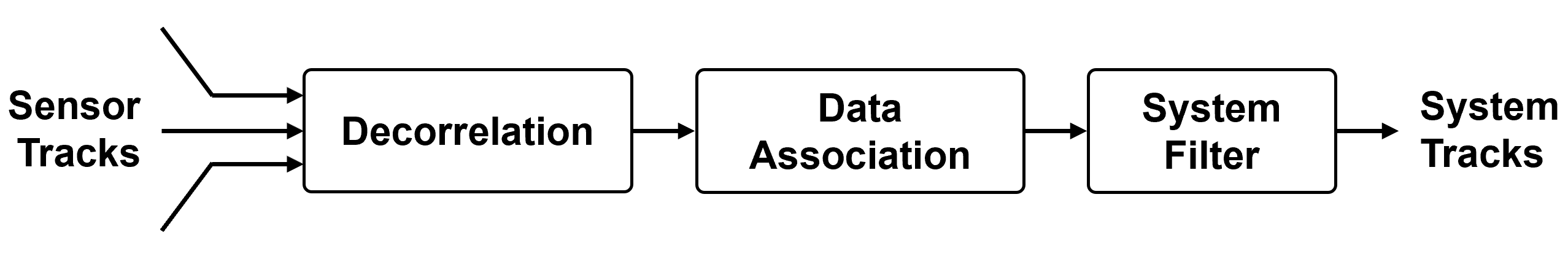}
  \caption{Multisensor, multitarget tracking with sensor tracks.}
\end{subfigure}
\caption{Multisensor tracking architectures.}\label{fig:arch}
\end{figure}

A major benefit of combining information over time and sensors is the ability to disambiguate the history of a track scene, i.e., retrodiction, through revised data association or more appropriate filter choices. For example, many real-time trackers that are used to control sensors employ process noise to be able to maintain custody on targets throughout unpredictable maneuvers. Upon acquiring a set of data, it may be clear that the target (via fitting or filtering metrics) was not maneuvering over the observed time span. In this case, the previous data history can be refiltered to create a tailored state estimate over the past data; this will ideally create an improved state estimate over the history and up to the current time. (Note that a tailored state estimate could also be approached via a multiple model filter~\cite{Mazor98, Hofbaur04}; however, these can still experience lag as compared to the benefit of hindsight in retrodiction and/or can be used in concert with retrodiction.) The major hurdle encountered with refiltering scene histories in many systems is a preliminary requirement of decorrelating the track information from the source sensors---this is the main focus of this paper. More specifically, the scope of this work is the investigation of temporal decorrelation of state estimates from a single source sensor to enable refiltering at a fusion node. It is assumed that the local measurement model is not known to the fusion node, but some aspects of the process noise model are available (discussed below).

Temporal decorrelation approaches can vary depending on the filtering assumptions, e.g., noise models, bias models, and how much information is known about the source sensors, e.g., measurement model. If the observed targets undergo ballistic or other deterministic motion, the filtering at the sensors do not require process noise. In this case, decorrelation of track states both across time and sensors was studied in~\cite{Frenkel95, Drummond95, Drummond02}. In practical applications, some level of process noise is typically required in tracking algorithms to compensate for potential maneuvers, unmodeled dynamics, or residual sensor measurement errors. The presence of process noise precludes the use of the methods in~\cite{Frenkel95, Drummond95, Drummond02} as it violates their founding assumptions and would generate decorrelated information with invalid covariances. To address this, this work proposes a decorrelation method that can account for process noise used in the inputted state estimates, whether the true process noise model is known at the central node or requires conservative estimation. 

The approach taken to develop the new decorrelation technique is similar to~\cite{Pao96}, which first takes a linear algebraic perspective to form effective measurement information from the incoming state estimates before passing to any filtering or fusion algorithms. However, to remove any need for knowing measurement mapping information from each constituent sensor, measurement reconstruction for the purposes of this work is constrained to exist in a common state space; this allows for a common decorrelation framework and enables the estimation of effective process noise at the central node. (Note that this is similar to the construction of a global information gain~\cite{Chong18}.) Thus, the main contribution of this work is to extend a measurement reconstruction methodology that exploits a common state space used at the fusion node and then construct a procedure for conservatively estimating process noise parameters used at contributing sensors. The effective measurement information with estimated process noise will then be used to perform refiltering of previous track data.

The paper is organized as follows. In Section~\ref{sec:ssem}, the formulation of effective measurement information from track states in a common state space is discussed, including the treatment of process noise used in a constituent sensor's filtering algorithm. Next, the case of unknown process noise is investigated in Section~\ref{sec:proc}, and a method of conservatively estimating process noise model parameters is given. To demonstrate the utility of the measurement reconstruction and process noise estimation techniques, simulations using a phased array radar model and ballistic target are given in Section~\ref{sec:sim}. Finally, overall conclusions and desired next steps are discussed in Section~\ref{sec:conc}.

\section{State Space Equivalent Measurements}\label{sec:ssem} 
To begin, consider a sensor observing a target while estimating the target's kinematic state.  The discrete-time kinematic state of the target $\bx_k \in \mathbb{R}^N$ evolves according to the following plant equation:
\[
\bx_{k} = f(\bx_{k-1}) + \bv_{k-1},
\]
\noindent where \(f(\cdot)\) is a known dynamics function from time $k-1$ to time $k$, and $\bv_{k-1}$ is additive process noise.  The process noise, $\bv_{k-1}$, is assumed to be white and Gaussian distributed such that $\bv_{k-1} \sim \cN(\bzeros, \bQ_{k-1})$.  The measurements, \(\bz_k \in \mathbb{R}^M\), at the sensor are of the form
\[
\bz_k = h(\bx_k) + \bw_k,
\]
\noindent where \(h(\cdot)\) is a measurement function known to the sensor and $\bw_k$ is measurement noise.  The measurement noise, $\bw_{k}$, is assumed to be white, independent of the process noise, and Gaussian distributed such that $\bw_{k} \sim \cN(\bzeros, \bR_k)$.

Due to the potential nonlinearity of the dynamics and measurement functions, it is assumed that the sensor employs an extended Kalman filter~\cite{BarShalom01} to track the kinematic state of the target.  In this case, the target state estimate at time $k$ given the data up until time $k$, $\hat{\bx}_{k\given k}$, and its covariance, $\bP_{k\given k}$, evolve according to the following equations:
\begin{subequations}\label{eq:kalman}
\begin{align}
\hat{\bx}_{k\given k-1} &= f\left(\hat{\bx}_{k-1\given k-1}\right),\\
\bP_{k\given k-1} &= \bF_{k}\bP_{k-1\given k-1}\bF_{k}^T + \bQ_{k-1},\label{eq:predp}\\
\hat{\bx}_{k\given k} &= \hat{\bx}_{k\given k-1} + \bK_{k}\left(\bz_{k} - h\left(\hat{\bx}_{k\given k-1}\right)\right),\label{eq:stateup}\\
\bP_{k\given k} &= \bP_{k\given k-1} - \bK_{k}\bS_{k}\bK_{k}^T,\label{eq:covup}\\
\bK_{k} &= \bP_{k\given k-1}\bH_{k}^T\bS^{-1}_{k},\\
\bS_{k} &= \bH_{k}\bP_{k\given k-1}\bH_{k}^T + \bR_{k}.
\end{align}
\end{subequations}
\noindent where \(\bF_{k}\) and \(\bH_{k}\) are Jacobians of the dynamics and measurement functions such that
\begin{align*}
  \bF_{k} &= \left.\frac{\partial}{\partial\bx}f\left(\bx\right)\right|_{\bx\hspace{0.5mm}=\hspace{0.5mm}\hat{\bx}_{k-1\given k-1}},\\
  \bH_{k} &= \left.\frac{\partial}{\partial\bx}h\left(\bx\right)\right|_{\bx\hspace{0.5mm}=\hspace{0.5mm}\hat{\bx}_{k\given k-1}}.
\end{align*}
\noindent Now, at the fusion node, the sensor's state estimates and covariances at each time instance, \(\hat{\bx}_{k\given k}\) and \(\bP_{k\given k}\), are obtained. The measurement model given by the measurement matrices, \(\bH_k\), and measurement covariance matrices, \(\bR_k\), are assumed to be unavailable to the fusion node, however. The dynamic model (and thus the matrices \(\bF_k\)) are assumed to be known to the fusion node.

In many cases, filtering may want to be performed again at a central data node using a different set of parameters, e.g., process noise, measurement noise; however, as the central node only has access to the state estimates and covariances from each sensor, this introduces the nuisance of reverse engineering effective measurement data from the history of state estimates~\cite{Pao96} (sometimes also called the construction of \emph{pseudomeasurements}). In order to accomplish this, it is desired to obtain quantities unknown to the fusion node, i.e., Kalman gain, measurement vector and covariance, in terms of the known state estimates and covariances. To begin, one can first analyze further the filter equations in (\ref{eq:kalman}).  Using the matrix inversion lemma~\cite{Golub96}, the covariance update equation (\ref{eq:covup}) can be rewritten in terms of inverses~\cite{BarShalom01} as
\begin{align}
\bP^{-1}_{k\given k} &= \bP^{-1}_{k\given k-1} + \bH_{k}^T\bR_{k}^{-1}\bH_{k}.\nonumber
\end{align}
\noindent Thus, the measurement covariance, $\bR_{k}$ and measurement matrix, $\bH_{k}$, satisfy the following equation in terms of the propagated past and present state covariances:
\begin{equation}
\bH_{k}^T\bR_{k}^{-1}\bH_{k} = \bP^{-1}_{k\given k} - \bP^{-1}_{k\given k-1}.\label{eq:info}
\end{equation}
\noindent Note that the term \(\bH_{k}^T\bR_{k}^{-1}\bH_{k}\) is often called the \emph{information gain} as it is the change from the propagated covariance to the updated covariance. The information gain will be denoted as \(\bJ_k\) such that
\[
\bJ_k = \bH_{k}^T\bR_{k}^{-1}\bH_{k}.  
\]
\noindent Continuing, the covariance update equation (\ref{eq:covup}) can also be rewritten as
\begin{align}
\bP_{k\given k} &= \left(\bI_N - \bK_{k}\bH_{k}\right)\bP_{k\given k-1}.\label{eq:kh}
\end{align}
\noindent where \(\bI_n\) is an \(n \times n\) identity matrix. Using the nonsingularity of the covariance \(\bP_{k\given k-1}\), the Kalman gain, $\bK_{k}$, and measurement matrix, $\bH_{k}$, obey
\begin{align}
\bK_{k}\bH_{k} &= \bI_N - \bP_{k\given k}\bP^{-1}_{k\given k-1}.\label{eq:k}
\end{align}
\noindent Finally, the state update equation (\ref{eq:stateup}) can be expressed as
\begin{align}
\hat{\bx}_{k\given k} &= \hat{\bx}_{k\given k-1} + \bK_{k}\left(\bz_{k} - h\left(\hat{\bx}_{k\given k-1}\right)\right),\nonumber\\
&= \hat{\bx}_{k\given k-1} - \bK_{k}h\left(\hat{\bx}_{k\given k-1}\right) + \bK_{k}\bz_{k},\label{eq:z1}
\end{align}
\noindent Rearranging (\ref{eq:z1}), one obtains
\begin{equation}
  \bK_{k}\bz_{k} = \hat{\bx}_{k\given k} - \hat{\bx}_{k\given k-1} + \bK_{k}h\left(\hat{\bx}_{k\given k-1}\right).\label{eq:z2}
\end{equation}
In summary, the three following equations are of interest:
\begin{align}
\bH_{k}^T\bR_{k}^{-1}\bH_{k} &= \bP^{-1}_{k\given k} - \bP^{-1}_{k\given k-1},\nonumber\\
\bK_{k}\bH_{k} &= \bI_N - \bP_{k\given k}\bP^{-1}_{k\given k-1},\nonumber
\end{align}
\begin{equation}
\bK_{k}\bz_{k} = \hat{\bx}_{k\given k} - \hat{\bx}_{k\given k-1} + \bK_{k}h\left(\hat{\bx}_{k\given k-1}\right),\nonumber
\end{equation}
\noindent where $\bP_{k\given k-1}$ implicitly depends on $\bF_k$, $\bQ_{k-1}$, and $\bP_{k-1\given k-1}$ as in (\ref{eq:predp}).

Assuming, for now, that the dynamics function, $f(\cdot)$, and the process noise covariances, $\bQ_{k-1}$, are known, the only unknowns in the equations above are $\bz_{k}$, $\bR_{k}$, $\bK_{k}$, and $\bH_{k}$.  In fact, the equations above only specify a subspace of the unknowns; it can be observed that the solutions for $\bR_{k}$ and $\bH_{k}$ are not unique.  This makes intuitive sense as it is possible for two different measurement mechanisms with appropriate uncertainties to garner the same information from the target state.  However, the original mapping from target space to measurement space is not of necessity to a central node.  Instead, it is possible to solve the equations for the equivalent measurement (or gain in information) in the same space as the target state---these reconstructed measurements will now be referred to as \emph{state space equivalent measurements} (SSEM).

To form the set of SSEM, the initial step is to assume a measurement function of the form
\begin{equation}
h\left(\bx_k\right) = \bH_k\bx_k,\label{eq:h}
\end{equation}
\noindent where
\begin{equation}
\bH_{k} = \left[\bI_M ~ \bzeros_{M \times (N-M)} \right].
\end{equation}
\noindent where $M$ denotes the dimension of the measurements and $\bzeros_{m \times n}$ is an $m \times n$ matrix of zeros. Note that the dimension of the measurement, $M$, may not be known by the fusion node but can be estimated using numerical methods; this is, however, beyond the scope of this paper.  Using (\ref{eq:h}), the equation (\ref{eq:info}) can be simplified as
\begin{equation}
\left[\begin{array}{cc}
\bR_{k}^{-1}&\bzeros_{M \times (N-M)}\\
\bzeros_{(N-M) \times M}&\bzeros_{(N-M) \times (N-M)}
\end{array}\right] = \bP^{-1}_{k\given k} - \bP^{-1}_{k\given k-1},\label{eq:info2}
\end{equation}
\noindent where it is important to note that \(\bR_{k}^{-1}\) is the inverse of the SSEM measurement covariance matrix, i.e., not the inverse of the original measurement covariance matrix at the sensor. 

Therefore, using (\ref{eq:info2}), one can solve for the SSEM covariance matrix, $\bR_{k}$.  Similarly, the Kalman gain can be obtained by noting that the equation (\ref{eq:k}) becomes
\begin{equation}
\left[\bK_{k}~\bzeros_{N \times (N-M)}\right] = \bI_N - \bP_{k\given k}\bP^{-1}_{k\given k-1}.\label{eq:gain}
\end{equation}
\noindent Revisiting the measurement equation in (\ref{eq:z2}) and substituting the measurement function from (\ref{eq:h}), the following is true:
\begin{align}
  \bK_{k}\bz_{k} &= \hat{\bx}_{k\given k} - \hat{\bx}_{k\given k-1} + \bK_{k}h\left(\hat{\bx}_{k\given k-1}\right),\nonumber\\
  &= \hat{\bx}_{k\given k} - \left(\bI_N - \bK_{k}\bH_{k}\right)\hat{\bx}_{k\given k-1},\nonumber\\
  &= \hat{\bx}_{k\given k} - \bP_{k\given k}\bP^{-1}_{k\given k-1}\hat{\bx}_{k\given k-1},\label{eq:kz}
\end{align}
\noindent where (\ref{eq:kz}) is a consequence of (\ref{eq:kh}). Note that if $M = N$, then $\bK_{k}$ is full rank, and $\bz_{k}$ can be found by using the inverse of $\bK_{k}$.  More commonly, however, not all state dimensions are observed and, thus, $M < N$.  In this circumstance, $\bK_{k}$ is a tall matrix and (\ref{eq:kz}) represents an overdetermined linear system.  Accordingly, $\bz_{k}$ can be found by using the Moore-Penrose pseudoinverse of $\bK_{k}$~\cite{Golub96}.  Therefore, in general,
\begin{equation}
\bz_{k} = \bK^\dagger_{k}\left(\hat{\bx}_{k\given k} - \bP_{k\given k}\bP^{-1}_{k\given k-1}\hat{\bx}_{k\given k-1}\right),\label{eq:z}
\end{equation}
\noindent where $\left(\cdot\right)^\dagger$ is the Moore-Penrose pseudoinverse. There is now a path to reconstructing the measurement covariance matrix via (\ref{eq:info2}), the Kalman gain matrix using (\ref{eq:gain}), and the equivalent measurement vector through (\ref{eq:z}); the complete SSEM procedure is summarized in Lemma~\ref{lemma:ps}.

\begin{lemma}\label{lemma:ps}
Given a dynamics function, \(f(\cdot)\), process noise matrices, $\bQ_{k-1}$, and the rank of the original measurements, $M$, the state space equivalent measurements from a series of state estimates and covariances can be derived by:
\begin{enumerate}
\item Construct propagated past state estimates, $\hat{\bx}_{k\given k-1}$, and covariances matrices, $\bP_{k\given k-1}$, by
\begin{align}
\hat{\bx}_{k\given k-1} &= f\left(\hat{\bx}_{k-1\given k-1}\right),\nonumber\\
\bP_{k\given k-1} &= \bF_{k}\bP_{k-1\given k-1}\bF_{k}^T + \bQ_{k-1},\nonumber
\end{align}
\item Obtain the inverse measurement covariance matrices, $\bR^{-1}_{k}$, using
\[
\left[\begin{array}{cc}
\bR_{k}^{-1}&\bzeros_{M \times (N-M)}\\
\bzeros_{(N-M) \times M}&\bzeros_{(N-M) \times (N-M)}
\end{array}\right] = \bP^{-1}_{k\given k} - \bP^{-1}_{k\given k-1},
\]
\item Calculate the Kalman gain matrices, $\bK_{k}$, with
\[
\left[\bK_{k}~\bzeros_{N \times (N-M)}\right] = \bI_N - \bP_{k\given k}\bP^{-1}_{k\given k-1},
\]
\item Create measurement vectors, $\bz_{k}$, from
\[
\bz_{k} = \bK^{\dagger}_{k}\left(\hat{\bx}_{k\given k} - \bP_{k\given k}\bP^{-1}_{k\given k-1}\hat{\bx}_{k\given k-1}\right).
\]
\end{enumerate}
\end{lemma}

\noindent Figure~\ref{fig:ssem1} shows an example of extracting SSEM from a set of track states on a ballistic target; the SSEM are plotted against the original measurements (converted into the state space coordinate frame).  Sensor measurements were taken in range-direction-cosine coordinates (described in detail in Section~\ref{sec:sim}) and the state space is the Earth-centered rotation coordinate system~\cite{Hofmann06}. Note that, in this example, the process noise model and its parameters are assumed to be known. It can be observed that the SSEM agree closely with the original detections.  In the following sections, the complexity of addressing unknown process noise parameters is discussed.

\begin{figure}[h!]
  \centering
  \includegraphics[width=0.99\columnwidth]{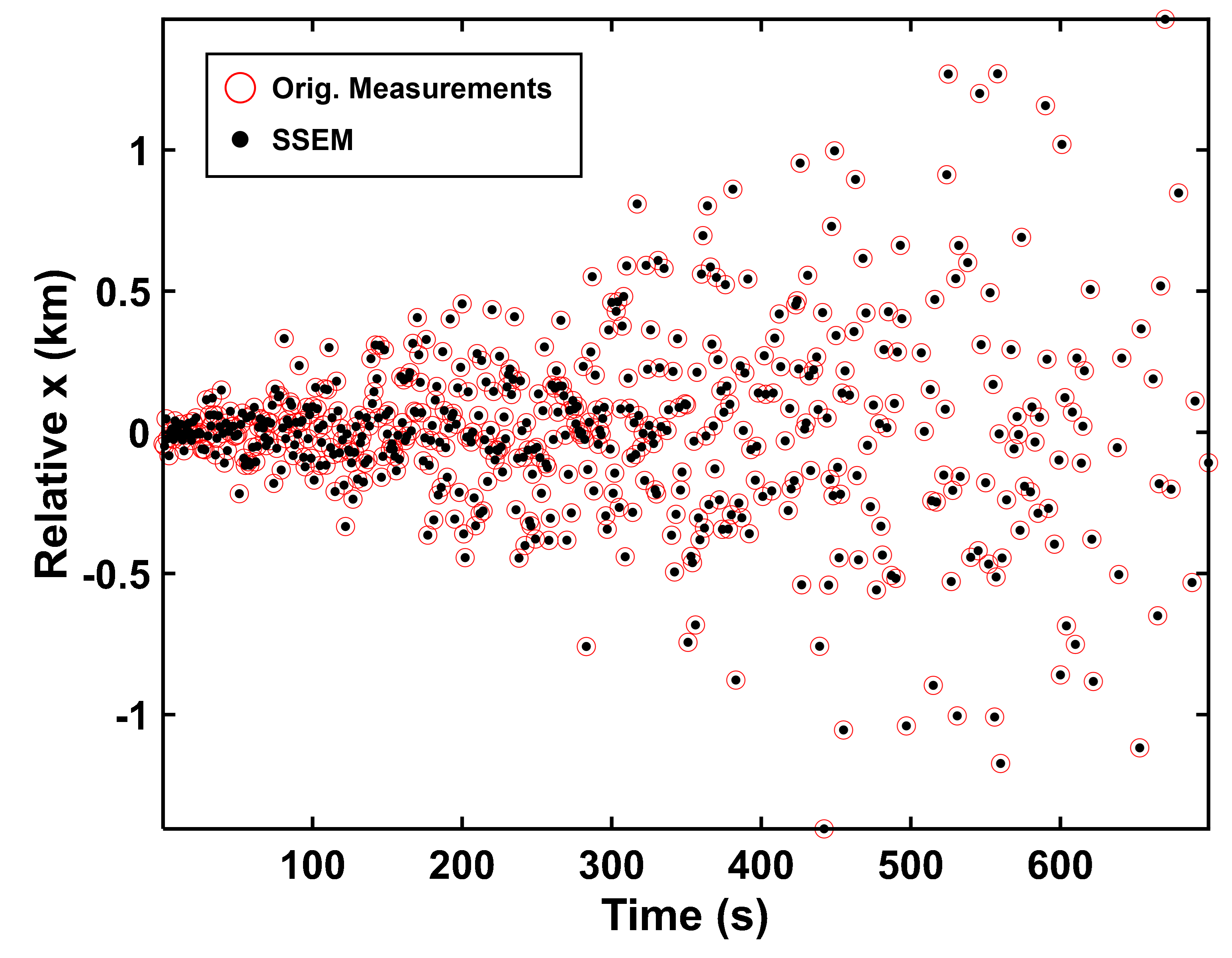}
  \caption{State space equivalent measurements derived from a set of state estimates and covariances constructed using range-direction-cosine measurements; measurements are displayed along the $x$-axis of the target space, Earth-centered rotating coordinates.}\label{fig:ssem1}
\end{figure}

\section{Unknown Process Noise}\label{sec:proc}
To further generalize the procedure in Lemma~\ref{lemma:ps}, it is now assumed that only limited structure of the process noise covariances, $\bQ_{k-1}$, is known, and the unknown aspects of the process noise matrices must be estimated. As a conservative measure, the process noise will be estimated in such a way as to upper bound the true SSEM measurement covariance, $\bR_k$, and, thus, construct a state estimate covariance that encapsulates the true SSEM state covariance. It is assumed that a parametric model of the process noise matrices \(\bQ_{k-1}(\btheta_{k-1})\) is known to the fusion node, but the parameter vector \(\btheta_{k-1}\) needs to be estimated; let \(\hat{\btheta}_{k-1}\) denote the fusion node's estimate of the process noise parameter vector.

Given a process noise parameter vector, \(\btheta\), the corresponding SSEM information gain matrix is defined as 
\[
\bJ_{k}(\btheta) = \bP^{-1}_{k\given k} - \left(\bF_{k}\bP_{k-1\given k-1}\bF_{k}^T + \bQ_{k-1}(\btheta)\right)^{-1}.
\]
\noindent Accordingly, the SSEM measurement covariance matrix for a given process noise parameter vector is 
\[
\bR_{k}(\btheta) = \bJ_{k}^{-1}(\btheta).
\]
Now, let the estimated process noise parameters be obtained using the following lemma:
\begin{lemma}\label{lemma:proc}
If the estimated process noise parameter vector, \(\hat{\btheta}\), is chosen such that
\begin{equation}
\bJ_k\left(\btheta'\right) \succeq \bJ_k\left(\hat{\btheta}\right) \succeq \bzeros,\label{eq:lemma2}
\end{equation}
\noindent for all values of \(\btheta' \neq \hat{\btheta}\), then $\bR_k\left(\hat{\btheta}\right) \succeq \bR_k\left(\btheta\right)$ where $\btheta$ is the true process noise parameter vector.
\end{lemma}
\begin{proof}
It is the goal to show that $\bR_k\left(\hat{\btheta}\right) \succeq \bR_k\left(\btheta\right)$; this is equivalent to showing $\bJ_k\left(\btheta\right) \succeq \bJ_k\left(\hat{\btheta}\right)$. This statement follows from the definition of \(\hat{\btheta}\) in the lemma as it is guaranteed to satisfy \(\bJ_k\left(\btheta'\right) \succeq \bJ_k\left(\hat{\btheta}\right)\) for all \(\btheta' \neq \hat{\btheta}\).
\end{proof}

\noindent Note that the approach proposed in Lemma~\ref{lemma:proc} is similar to applying a minimum encapsulating covariance methodology across the inverse state covariance matrices~\cite{Julier97}; in short, one is obtaining the smallest information gain that satisfies the matrix properties of the filtering equations~(\ref{eq:kalman}). In the following section, the procedure given in Lemma~\ref{lemma:proc} for estimating the process noise parameter vector is illustrated for a spherically-distributed, white Gaussian process noise of unknown power spectral density.

\section{Numerical Example}\label{sec:sim}

To illustrate the use of SSEM with process noise estimation, a series of ballistic target tracking simulations will now be exhibited. The discrete-time dynamic state, \(\bx_k\), is assumed to be three-dimensional and contain the position and velocity components, i.e., \(N = 6\), such that
\[
\bx_k = \left[\begin{array}{c}
  x_k\\
  y_k\\
  z_k\\
  \dot{x}_k\\
  \dot{y}_k\\
  \dot{z}_k
\end{array}\right],  
\]
\noindent where \(x_k, y_k, z_k\) are the three dimensions for position at the \(k^{\mathrm{th}}\) time instance and \(\dot{x}_k\) denotes the first time derivative of the first dimension. 

For the assumed process noise model, a spherically-distributed, white Gaussian process noise derived from a continuous-time process is used.  In this case, the process noise is assumed to originate from a zero-mean vector signal, $\bv(t)$, with autocovariance function given as
\begin{align}
\bV(t, \tau) &= E\left[\bv(t)\bv^T(\tau)\right],\nonumber\\
&= \eta(t)\delta(t - \tau)\bA,\nonumber
\end{align}
\noindent where $\eta(t) > 0$ is the unknown variance parameter of the process noise, $\delta(t)$ is the Dirac delta function, and
\[
\bA = \left[
\begin{array}{cc}
\bzeros_{3\times 3}&\bzeros_{3\times 3}\\
\bzeros_{3\times 3}&\bI_{3}\end{array}
\right].
\]
\noindent Note that the identity matrix in \(\bA\) is mapping the process noise into the acceleration dimension. Further, the variance parameter, $\eta(t)$, is assumed to be piecewise constant such that
\[
\eta(t) = \eta_k\textnormal{ for }t_{k-1}\leq t < t_k.
\]
\noindent In this case, the process noise covariances, $\bQ_{k-1}$, can be written as~\cite{BarShalom01}
\begin{align}
\bQ_{k-1} &= \displaystyle\int_{t_{k-2}}^{t_{k-1}}\bPhi(t_{k-1}, \tau)\bV(\tau, \tau)\bPhi^T(t_{k-1}, \tau)d\tau,\nonumber\\
&= \eta_{k-1}\displaystyle\int_{t_{k-2}}^{t_{k-1}}\bPhi(t_{k-1}, \tau)\bA\bPhi^T(t_{k-1}, \tau)d\tau,\nonumber\\
&= \eta_{k-1}\bB_{k-1},\nonumber
\end{align}
\noindent where
\[
\bB_{k-1} = \displaystyle\int_{t_{k-2}}^{t_{k-1}}\bPhi(t_{k-1}, \tau)\bA\bPhi^T(t_{k-1}, \tau)d\tau,
\]
\noindent is known to the fusion node as the dynamic model and, thus, the state transition matrices, $\bPhi(t, \tau)$, are available. Finally, as the instantaneous values of \(\eta_{k-1}\) are unknown, the process noise parameter vector is now a scalar given by \(\theta_{k-1} = \eta_{k-1}\) and \(\bQ_{k-1}(\theta_{k-1}) = \theta_{k-1} \bB_{k-1}\) where \(\bB_{k-1}\) is known. The estimation procedure from Lemma~\ref{lemma:proc} will be used to form the estimate of \(\theta_{k-1}\), i.e., \(\theta_{k-1}\) is found via a line search to satisfy (\ref{eq:lemma2}).

For this example, the observed target is undergoing pure ballistic motion and state measurements are obtained using a radar-like model. More specifically, the measurements are in range-direction-cosine (RUV) coordinates, commonly used as the measurement basis for phased array radar systems. The measurement function (i.e., conversion to RUV space) can be summarized as
\[
h\left(\bx_k\right) = \left[\begin{array}{c}
\sqrt{x^2 + y^2 + z^2}\\
\frac{x}{\sqrt{x^2 + y^2 + z^2}}\\
\frac{y}{\sqrt{x^2 + y^2 + z^2}}
\end{array}\right],
\]
\noindent where $x$, $y$, and $z$ are components in a sensor-centered Cartesian coordinate system with the \(z\)-axis pointed toward the target, and the output dimensions are the range, $u$, and $v$ dimensions, respectively. It can be observed that the measurement dimension is \(M = 3\).

The radar measurements are characterized by a covariance matrix \(\bR_k\) of the form
\[
  \bR_k = \left[\begin{array}{c c c}
    \sigma_r^2 & 0 & 0\\
    0 & \sigma_u^2 & 0\\
    0 & 0 & \sigma_v^2
  \end{array}\right],
\]
\noindent where \(\sigma_r\), \(\sigma_u\), \(\sigma_v\) are the standard deviations in the range, \(u\), and \(v\) dimensions, respectively. The measurement standard deviation in range, $\sigma_r$, is given by 
\[
  \sigma_r = \frac{c}{2B\sqrt{\rho}},  
\]
\noindent where \(c\) is the speed of light, \(B\) is the radar waveform bandwidth, and \(\rho\) is the signal-to-noise ratio (SNR)~\cite{Skolnik1960}. The SNR for a target at range \(r\) with radar cross-section (RCS) of \(\gamma\) calculated relative to a reference SNR, \(\rho_0\), for a target at range \(r_0\) and radar cross-section of \(\gamma_0\) as
\[
  \rho = \rho_0\left(\frac{\gamma}{\gamma_0}\right)\left(\frac{r_0}{r}\right)^4.
\] 
\noindent Similarly, the measurement standard deviation in the \(u\) and \(v\) dimensions, \(\sigma_u\) and \(\sigma_v\), are assumed to be equal and given by
\[
  \sigma_u = \sigma_v = \frac{\phi}{k_m\sqrt{\rho}},  
\]
\noindent where \(\phi\) is the 3-dB beamwidth of the radar and \(k_m\) is the angular error signal slope for the radar~\cite{Skolnik1960}. Note that the scenario has been chosen to use low-to-moderate SNR, and, therefore, limiting values at high SNR values for the uncertainties are not imposed. More specifically, for this example, the following values for the radar measurements have been chosen:
\begin{itemize}
  \item Radar bandwidth, \(B = 100\)~MHz
  \item Radar beamwidth, \(\phi = 1\)~mrad
  \item Reference SNR, \(\rho_0 = 0\)~dB, for a target of radar cross section \(\gamma_0 = 0\)~dBsm, at a range of \(r_0 = 2700\)~km
  \item Angular error slope, \(k_m = 1.6\)
\end{itemize}
\noindent The ballistic trajectory originates from a notional location at \(2^\circ\) latitude, \(5^\circ\) longitude, and \(0\)~m altitude. It is in flight for 700 seconds and then impacts at another notional point at \(10^\circ\) latitude, \(10^\circ\) longitude, and \(0\)~m altitude. The radar is location at \(0^\circ\) latitude, \(0^\circ\) longitude, and \(0\)~m altitude.

First, a simulation is provided to show using SSEM with estimated process noise can achieve comparable results to having access to the original source sensor measurements. In Figure~\ref{fig:sim1}, position and velocity error comparisons are given between two different filtering mechanisms: (i) an EKF run on the original radar measurements with white, spherically-distributed process noise with power spectral density \(\eta=0.01~\mathrm{m}^2/\mathrm{s}^3\), and (ii) a set of SSEM derived from the state estimates of (i) and then passed through an EKF with the same process noise model as (i). It can be seen that both in position and velocity, the error of both sets of state estimates agree very closely; there is only a slight difference in the very beginning of the track while the estimate is transitioning out of initialization. This shows that despite the SSEM being created from only track state information with estimated process noise, the filtered results are extremely close to using the original measurements.

\begin{figure}[h!]
\begin{subfigure}[b]{\columnwidth}
  \centering
  \includegraphics[width=0.9\columnwidth]{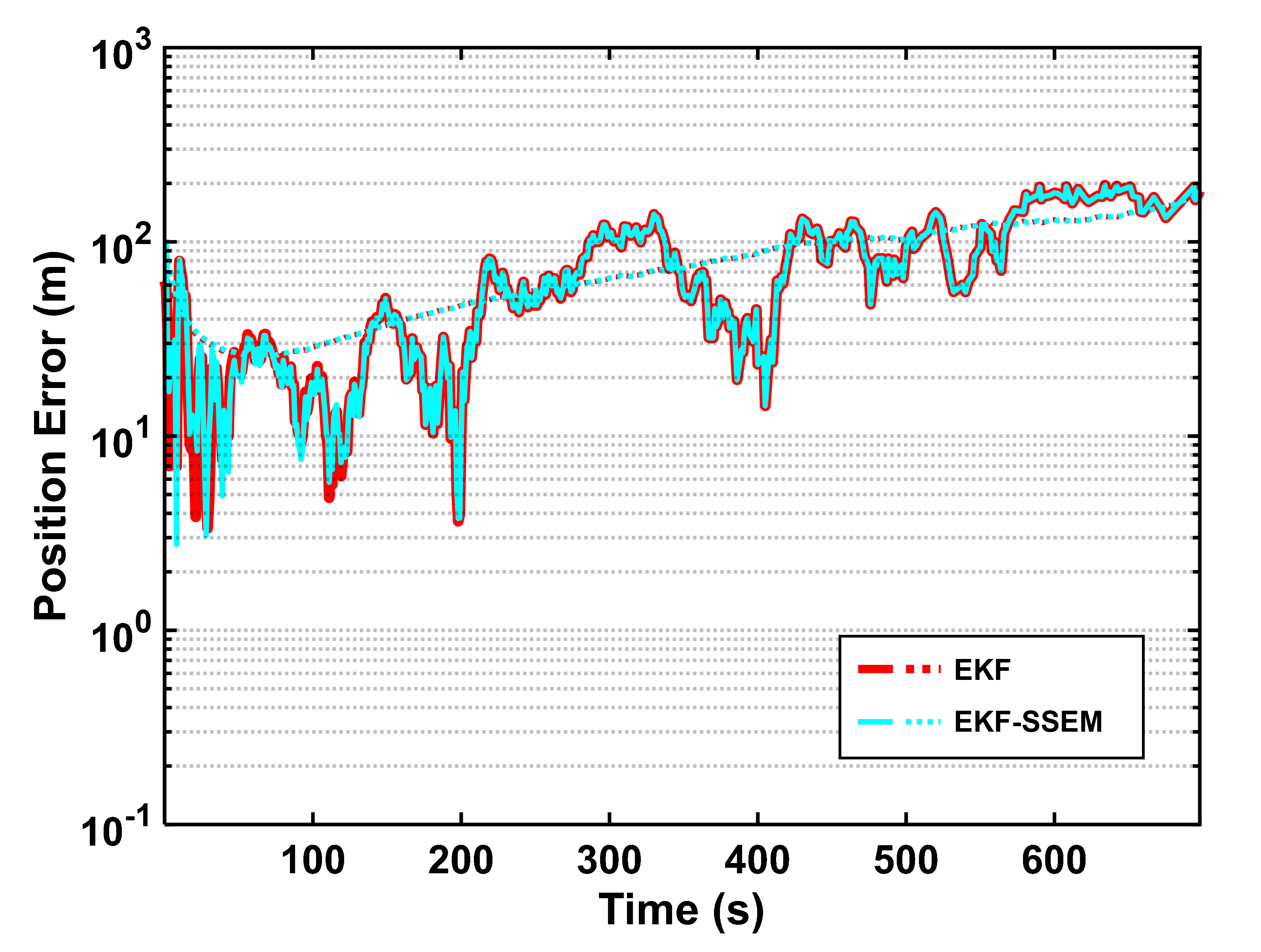}
  \caption{Position error comparison.}\label{fig:sim1_pos}
\end{subfigure}

\begin{subfigure}[b]{\columnwidth}
  \centering
  \includegraphics[width=0.9\columnwidth]{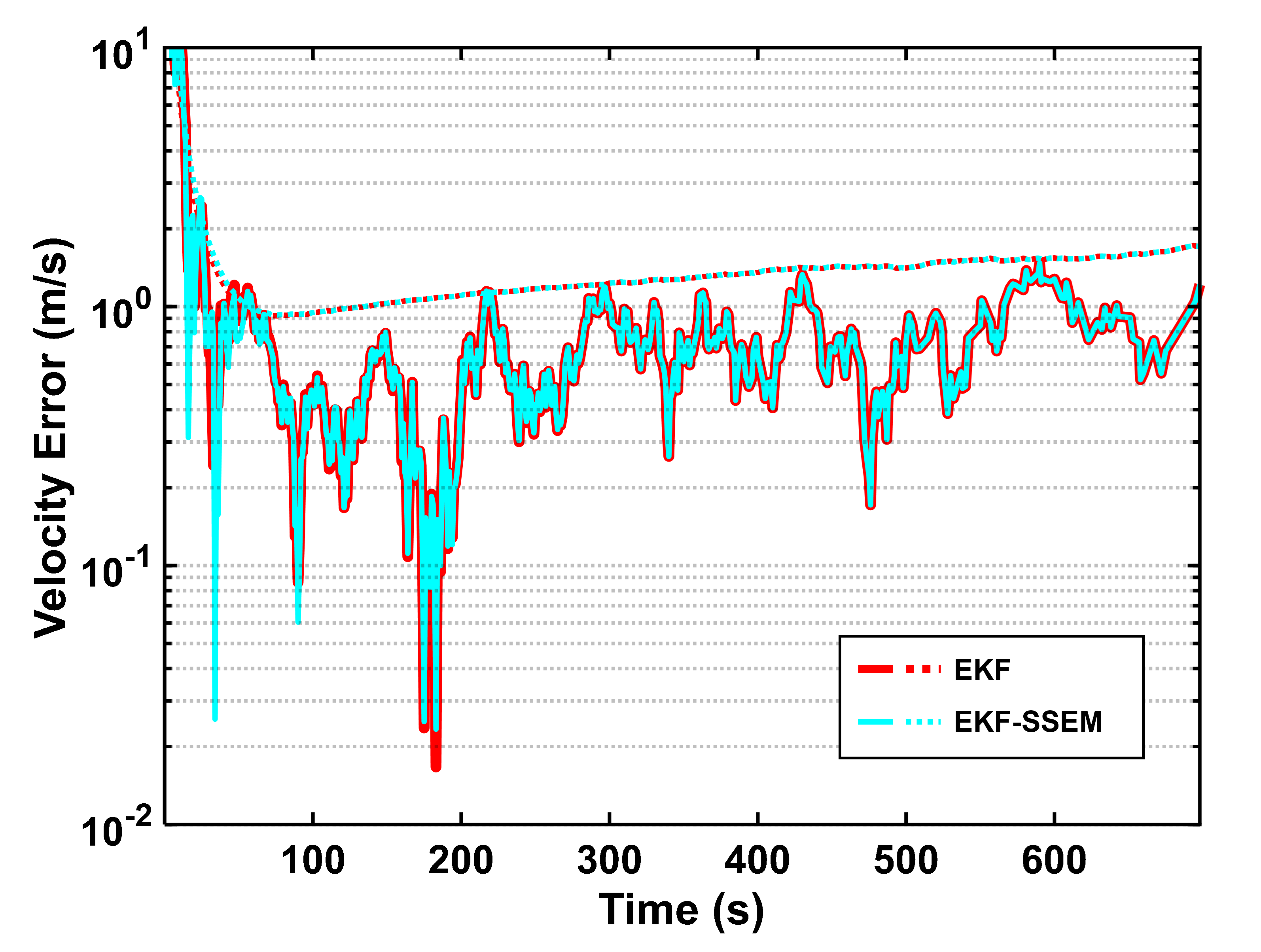}
  \caption{Velocity error comparison.}\label{fig:sim1_vel}
\end{subfigure}

\caption{Error comparison for two different filtering methods: EKF using radar measurements and an EKF using SSEM and estimated process noise. Both filters use a process noise power spectral density of \(\eta=0.01~\mathrm{m}^2/\mathrm{s}^3\). Error is shown as solid line, 1-\(\sigma\) reported error is shown as dashed line.}\label{fig:sim1}
\end{figure}

A case illustrating the benefits of refiltering source track states at a fusion node with different process noise model parameters than the originating source track filter is shown in Figure~\ref{fig:sim2}; this example compares the following two filtering results: (i) an EKF run on the original radar measurements with white, spherically-distributed process noise with power spectral density \(\eta=0.01~\mathrm{m}^2/\mathrm{s}^3\) (i.e., the original source tracks), and (ii) a set of SSEM derived from the original track states run through an EKF with a much lower process noise of \(\eta=0.00001~\mathrm{m}^2/\mathrm{s}^3\) (i.e., the refiltered tracks at the fusion node). In the figures, the main benefit of identifying the true motion mode via retrodiction and refiltering can be seen. By observing the absence of maneuvers over the course of the trajectory and applying a lower process noise, the refiltered estimate at the fusion node achieves significantly higher position and velocity accuracy.

\begin{figure}[h!]
  \begin{subfigure}[b]{\columnwidth}
    \centering
    \includegraphics[width=0.9\columnwidth]{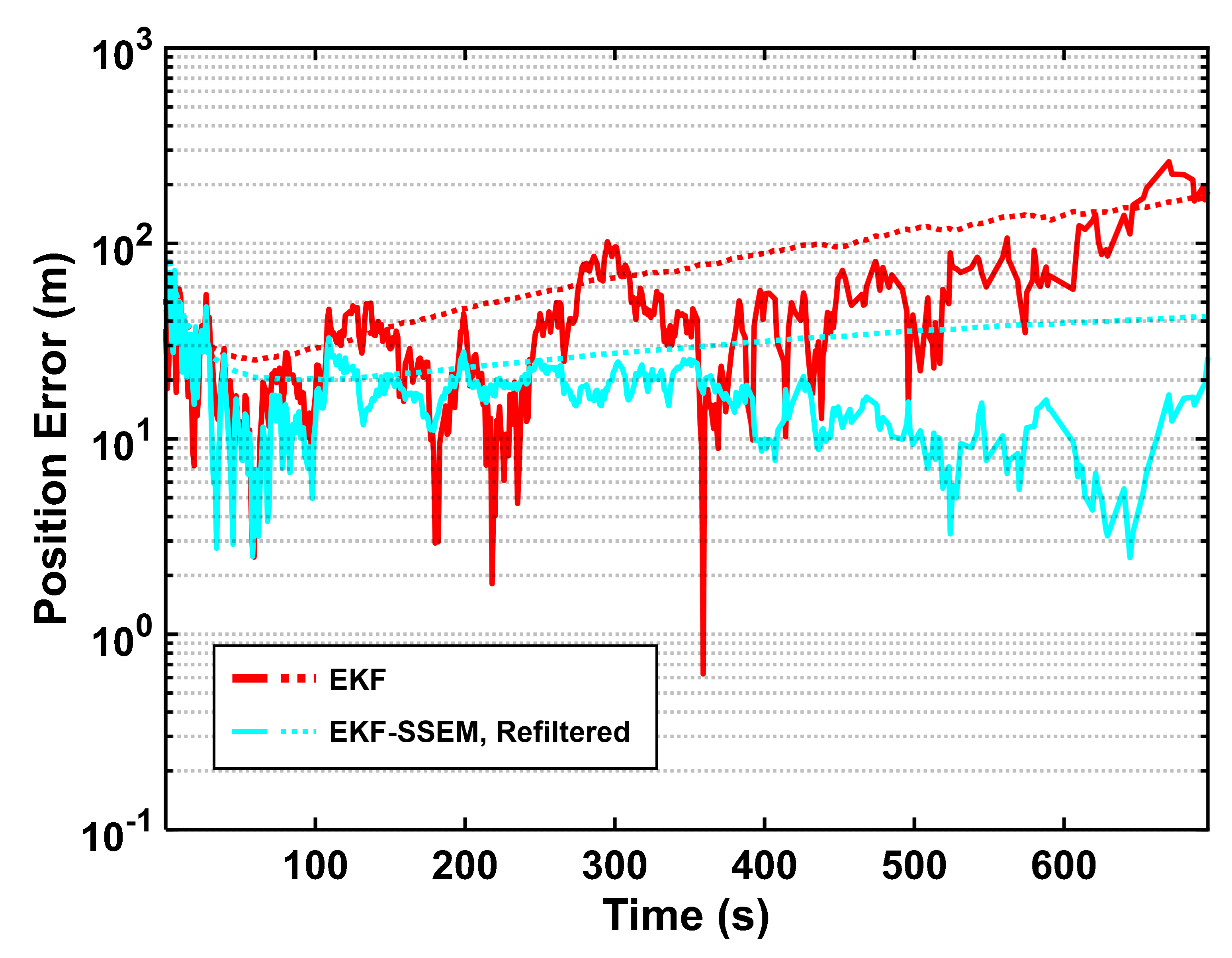}
    \caption{Position error comparison.}\label{fig:sim2_pos}
  \end{subfigure}
  
  \begin{subfigure}[b]{\columnwidth}
    \centering
    \includegraphics[width=0.9\columnwidth]{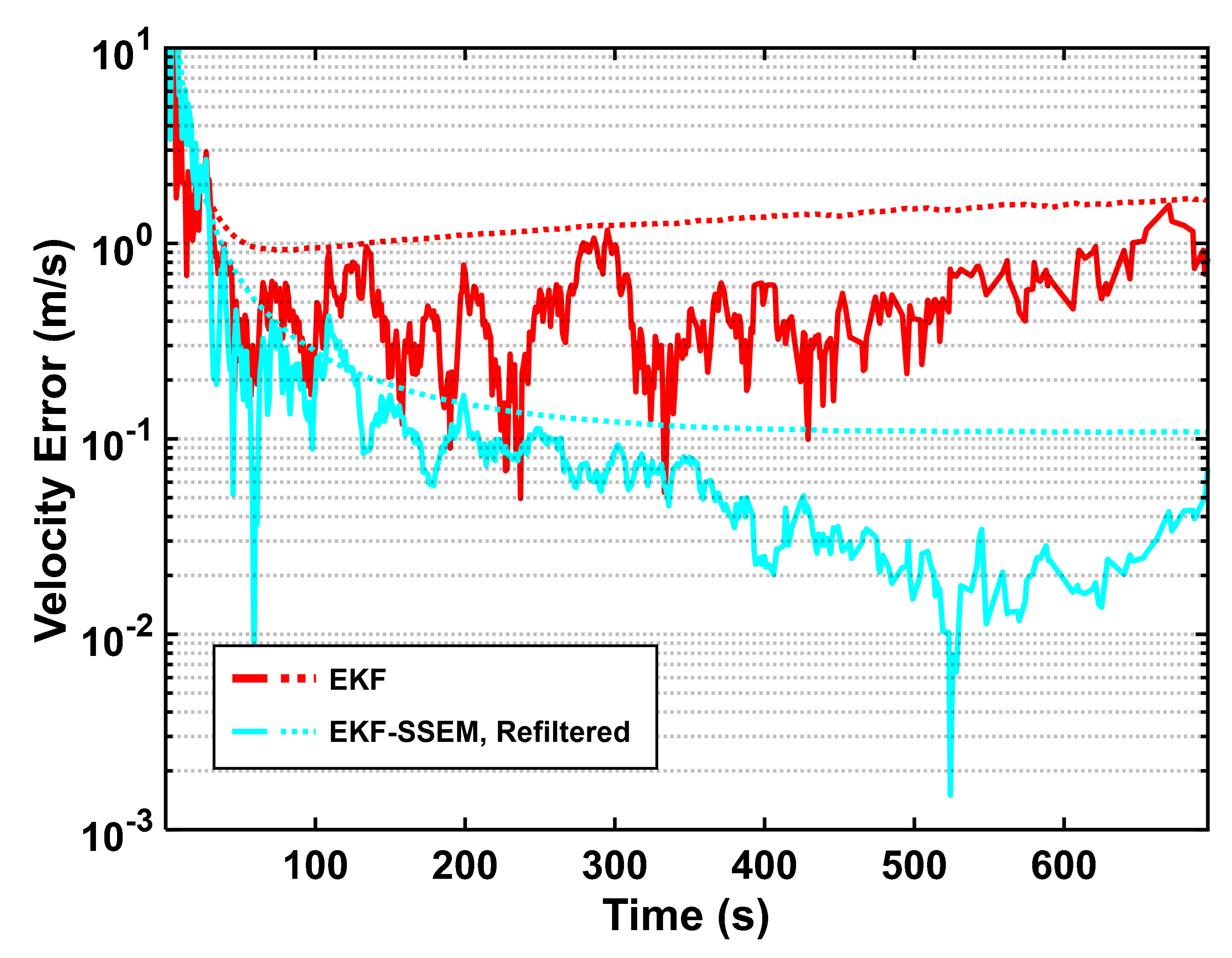}
    \caption{Velocity error comparison.}\label{fig:sim2_vel}
  \end{subfigure}
  
  \caption{Error comparison for two different filtering methods: EKF using radar measurements and the refiltered result using EKF with SSEM and estimated process noise. The former uses a process noise power spectral density of \(\eta=0.01~\mathrm{m}^2/\mathrm{s}^3\), and the latter uses power spectral density of \(\eta=0.00001~\mathrm{m}^2/\mathrm{s}^3\). Error is shown as solid line, 1-\(\sigma\) reported error is shown as dashed line.}\label{fig:sim2}
  \end{figure}

\section{Conclusions and Future Work}\label{sec:conc}

In this work, the temporal decorrelation of sequential state estimates is studied to enable refiltering at a track fusion node, where sensor measurement models and certain aspects of process noise models are unknown to the fusion node. The proposed technique is a generalization of the temporal decorrelation techniques proposed in~\cite{Frenkel95,Drummond95} using effective measurement reconstruction~\cite{Pao96} and accounting for process noise used in the inputted state estimates. Using this setup, adaptation to unknown process noise is developed to ensure a conservative state estimate at the central node. Results are then given to show the efficacy of the reconstructed measurements and process noise estimation. 

To further extend the proposed methods, it is also a goal to show that the repeated measurements between reported state estimates can be accounted for in the decorrelation process. This is important as track information is often reported to a central node at a lower rate than the native measurement rate. Using downsampled track states have been shown to work empirically, with a similar trade-off incurred as tracklet fusion with process noise and a mismatch in fusion and measurement rates~\cite{Chong14}, i.e., as the fusion rate decreases relative to measurement rate, an optimality gap is observed; a formal discussion is desired for subsequent work.

Finally, the author would like to thank the reviewers for their valuable input, which improved the paper and broadened the relevant references.

\bibliographystyle{IEEEtran}
\bibliography{pseudo.bib}

\end{document}